\documentclass[runningheads]{llncs}

\usepackage{graphicx}
\usepackage{tikz}
\usepackage{comment}

\usepackage{caption}
\usepackage{subcaption}

\usepackage{enumitem}
\usepackage{amsfonts}
\usepackage{amssymb}
\usepackage{amsmath}

\begin{document}

\newcommand{\st}[1]{\mathbf{#1}}
\newcommand{\car}{\circlearrowright}
\newcommand{\q}{\textbf{q}}
\newcommand{\p}{\textbf{p}}
\newcommand{\s}{\textbf{s}}
\newcommand{\f}{\textbf{f}}
\newcommand{\Ra}{\rightarrow}
\newcommand{\bigo}{{\mathcal O}}

\newcommand{\N}{\mathbb{N}}
\newcommand{\swe}{\textrm{SWEEP}}

\newcommand{\marked}{\varepsilon}

\title{Sweep Complexity Revisited\thanks{This work was supported by JSPS KAKENHI Grant Number JP23K10976.}}

\titlerunning{Sweep Complexity Revisited}

\author{Szil\'ard Zsolt Fazekas\inst{1}\orcidID{0001-5319-0395} \and
Robert Merca\c{s}\inst{2}\orcidID{0001-6034-433X}}

\authorrunning{S.Z. Fazekas and R. Merca\c{s}}

\institute{Akita University, Department of Mathematical Science and Electrical-Electronic-Computer Engineering\\
\email{szilard.fazekas@ie.akita-u.ac.jp} \and
Loughborough University, Department of Computer Science\\
\email{R.G.Mercas@lboro.ac.uk}}

\maketitle              

\begin{abstract}
We study the sweep complexity of DFA in one-way jumping mode answering several questions posed earlier. This measure is the number of times in the worst case that such machines have to return to the beginning of their input after having skipped some of the symbols. The class of languages accepted by these machines strictly includes the regular class and constant sweep complexity allows exactly the acceptance of regular languages. However, we show that there exist machines with higher than constant complexity still only accepting regular languages and that in general the sweep complexity of an automaton does not distinguish between accepting regular and non-regular languages. We establish separation results for asymptotic classes defined by this complexity measure and give a surprising exponential/logarithmic relation between factors of certain inputs which can be verified by such machines.

\end{abstract}

\section{Introduction}
In roughly the last three decades, several non-classical models of automata have been introduced to study the effect of processing inputs with simple machines in a non-sequential way. Such models include restarting automata~\cite{restart}, jumping automata~\cite{JFA}, input revolving automata~\cite{revolving} and automata with translucent letters~\cite{translucent}. However, these models are either strictly more powerful or accept a class incomparable with the regular one.

One-way jumping finite automata (OWJFA) were introduced~\cite{chigahara} to study the power of deterministic finite automata (DFA) performing non-sequential processing without completely discarding structural information about the inputs \`a la jumping automata. The resulting model is, in a sense, a minimal extension of finite automata. Machines are specified in exactly the same way as DFA allowing partial transition functions. The only change is the behaviour of the machine when encountering a letter for which the current state has no outgoing transition defined. In the classical case such inputs are rejected, but in one-way jumping mode the letters are skipped temporarily to be processed later. The relative order of the skipped symbols is maintained, and the automaton moves back to the beginning after each pass (called \emph{sweep} here), seeing only the symbols previously skipped. Therefore one can also view this model as a DFA with an input tape which works as a restricted queue, or one that reads and erases symbols from a circular tape always jumping clockwise to the nearest letter for which it has a defined transition from the current state. When the transition function is complete, no symbols are skipped, so the machine behaves as ordinary DFA, which means that the class of \mbox{languages} accepted by DFA in one-way jumping mode trivially includes all regular languages.

Various properties of the accepted language class~\cite{ROWJFA} and the status of fundamental decidability questions have been settled~\cite{DJFA}. More powerful machines with this new processing mode have also been investigated, such as nondeterministic finite automata~\cite{BeierHolzerNFA,FHY21natcomp}, two-way finite automata~\cite{FHY21tcs}, pushdown automata and linear bounded automata~\cite{FHY21natcomp}. 
While the language classes defined by the models have no nontrivial closure properties under usual language operations, the accepting power and decidability issues raised some intriguing problems. 

Except for linear bounded automata, the machine models mentioned above become more powerful when they are allowed to jump to the nearest symbol readable in the current state, which is not surprising. However, it has proven challenging to get a clear picture of just how powerful the new processing mode is, even in the simplest case when one starts from DFA. 
Such automata can accept all regular languages and the language class defined by them is incomparable with the context-free class, but included in the context-sensitive class and in DTIME($n^2$)~\cite{ROWJFA}. The separation results make use of combinations of a handful of regular languages together with a very simple type of non-regular languages which contain words having letter counts in a certain ratio, e.g., the frequently used $L_{ab}=\{w\in\{a,b\}^* \mid w \textrm{ contains as many }a \textrm{'s as }b \textrm{'s}\}$ accepted by the machine $\mathcal{A}$ in Fig~\ref{fig:two} (with states $\st{1}$, or $\st{2}$ final). 
While this was enough to establish virtually all separations of interest, it left a significant gap in our understanding of the model: can such machines accept any (`interesting') non-regular languages apart from the ones which establish linear relationships among letter counts?

In this work we answer the question above, building on the investigation of sweep complexity of DFA in one-way jumping mode. Sweep count can be viewed as a measure of non-regular resources used by a machine posing the natural question of how much of this resource is needed to be able to accept non-regular languages? It has been shown that constant sweep complexity does not increase the accepting power of the machines~\cite{NCMA16} and that superconstant sweep complexity requires cycles containing `complementary deficient' states~\cite{FMW22}. In the latter paper it was conjectured that, in fact, any automaton with higher than constant sweep complexity accepts a non-regular language. In Section~\ref{sec:logregular} we refute that conjecture by exhibiting a small DFA accepting a regular language while processing some inputs of length $n$ in $\Omega(\log n)$ sweeps. We also show that there is no non-trivial upper bound on the sweep complexity of regular languages, that is, there are machines with linear complexity accepting regular languages. 

A natural question regarding the new complexity measure is whether there exists a meaningful hierarchy which does not collapse to the extremes of $\mathcal{O}(1)$ and $\mathcal{O}(n)$. The aforementioned example shows that automata with logarithmic complexity exist, which answers another question posed earlier. Furthermore, following the line of computational complexity theory, we set out to explore whether the language classes defined through asymptotic complexity form a true hierarchy, that is whether there are languages which can be accepted by a machine with $\mathcal{O}(f(n))$ complexity but not by any with $o(f(n))$ complexity, for various functions $f(n)$. In Section~\ref{sec:lognonregular} we demonstrate that such a hierarchy exists by presenting languages with $\Theta(\log n)$ and $\Theta(n)$ sweep complexity, respectively.

Finally we mention that sweep complexity as an idea has been studied in other contexts, too: an interesting and thorough investigation of a similar flavor established infinite hierarchies in terms of sweep count for iterated uniform finite transducers~\cite{iteratedfst2022}, although that model is significantly more powerful than ours, so the techniques used there do not translate here as far as we can tell.

\section{Preliminaries}
We consider words over a finite alphabet, e.g., $\Sigma=\{a,b\}$. The set of all words over $\Sigma$ is $\Sigma^*$, which includes the empty word $\varepsilon$.

A \textit{DFA} is a quintuple $M=(Q,\Sigma,R,\textbf{s},F)$, where $Q$ is the finite set of states, $\Sigma$ is the finite input alphabet, 
$\Sigma \cap Q = \emptyset$, $R: Q \times \Sigma \rightarrow Q$ is the transition function, 
$\textbf{s} \in Q$ is the start state, and $F\subseteq Q$ is the set of final states. Elements of $R$ are referred to as (transition) rules of $M$ and we write 
$\textbf{p}y \rightarrow \textbf{q} \in R$ instead of $R(\textbf{p}, y)= \textbf{q}$. A configuration of $M$ is a string in $Q \times \Sigma^{*}$. 

A DFA transitions from a configuration $\textbf{p}w$ to a configuration $\textbf{q}w'$ if $w=aw'$ and $\textbf{p}a\rightarrow \textbf{q}\in R$, with $\p,\q\in Q$, $w,w'\in \Sigma^*$ and $a\in \Sigma$. By extending the meaning of $\Ra$ we denote this by $\p w\Ra \q w'$ and the reflexive and transitive closure of $\Ra$ by $\Ra^*$. A word $w$ is \textit{accepted} by a DFA $M$ if there exists $\textbf{f}\in F$, such that $\textbf{s}w\Ra^* \textbf{f}$. The language accepted by $M$ is $\{w\in \Sigma^* \mid \exists \textbf{f}\in F: \textbf{s}w\Ra^* \textbf{f}\}$.

\begin{figure}[b]
\begin{minipage}{.45\textwidth}
\centering
        \begin{center}
\begin{tikzpicture}[scale=0.2]
\tikzstyle{every node}+=[inner sep=0pt]
\draw [black] (13,-28.7) circle (3);
\draw (13,-28.7) node {$\st{1}$};
\draw [black] (30.6,-28.7) circle (3);
\draw (30.6,-28.7) node {$\st{2}$};
\draw [black] (15.684,-27.368) arc (111.29246:68.70754:16.843);
\fill [black] (27.92,-27.37) -- (27.35,-26.61) -- (26.99,-27.54);
\draw (21.8,-25.72) node [above] {$a$};
\draw [black] (27.808,-29.79) arc (-72.87649:-107.12351:20.406);
\fill [black] (15.79,-29.79) -- (16.41,-30.5) -- (16.7,-29.55);
\draw (21.8,-31.2) node [below] {$b$};
\end{tikzpicture}
\end{center}
        \caption{The only two-state ROWJFA satisfying Lemma~\ref{nonconstantsweeps}}
        \label{fig:two}
\end{minipage}
\hfill
\begin{minipage}{.52\textwidth}
    \centering
        $
\begin{array}{r|ccccccc}
      position: & \mbox{ } 0 & \mbox{ } 1 & \mbox{ } 2 & \mbox{ } 3 & \mbox{ } 4 & \mbox{ } 5 & \mbox{ } 6 \\
     \hline
     \color{red} \textrm{input} & \color{red} \bf{a} & \color{red} \bf{d} & \color{red} \bf{c} & \color{red} \bf{b} & \color{red} \bf{c} & \color{red} \bf{b} & \color{red} \bf{a} \\
     \textrm{after sweep } 1 & \marked & \bf{d} & \bf{c} & \bf{b} & \bf{c} & \bf{b} & \marked \\
     \textrm{after sweep } 2 & \marked & \bf{d} & \bf{c} & \marked & \bf{c} & \marked & \marked\\
     \textrm{after sweep } 3 & \marked & \bf{d} & \marked & \marked & \marked & \marked & \marked\\
     \textrm{after sweep } 4 & \marked & \marked & \marked & \marked & \marked & \marked & \marked
\end{array}
$
\caption{The computation table for $adcbcba$ by the machine in Example~\ref{example1}.}
\label{fig:comp_table}
\end{minipage}
\end{figure}

\noindent
{\bf One-way jumping automata}\\
The \textit{one-way jumping relation} (denoted by $\circlearrowright$) 
between configurations from $Q \Sigma^{*}$, was originally defined in~\cite{chigahara}. Here we follow the slightly different definition of \cite{FMW22} which does not change the accepting power of the model, but is more convenient.

A tuple $M=(Q,\Sigma, R,\s,F)$ representing a \emph{deterministic right one-way jumping automaton} (ROWJFA) is  defined the same way as a DFA, where the configurations are also elements of the set $Q\times \Sigma^*$. Let $\Sigma_p = \{b\in \Sigma \mid \exists \q\in Q$  such that  $\p b\rightarrow \q \in R \}$ be the set of all of the letters from $\Sigma$ for which we have a  transition defined from state $\p$.
A jumping transition (or jump, for short), denoted $\circlearrowright$, is defined between configurations $\p ax$ and $\p xa$ if state $\p$ cannot read the letter $a$, formally:
\[\p ax \circlearrowright \p xa, \mbox{ if } a\in \Sigma \setminus \Sigma_p.\]

A ROWJFA can transition from configuration $\textbf{p}ax$ to configuration $\textbf{q}y$, which we denote by $\p ax\vdash \q y$, if 
\begin{eqnarray*}
(i) && \p ax \Ra \q y, \mbox{ where } x=y \mbox{ and } \p a\rightarrow \q\in R, \mbox{ as defined earlier, or}\\
(ii) && \p ax \circlearrowright \p xa,\mbox{ when } a\in \Sigma \setminus \Sigma_p, \p=\q \mbox{ and } xa=y.
\end{eqnarray*}

 A word $w$ is accepted by $M$ if $\s w\vdash^* \f$. The language accepted by $M$ is defined by $L(M)=\{x\in\Sigma^*\mid \exists \f\in F: \s x\vdash^* \f \}.$

While some texts define DFA having complete transition functions, our DFA allow partially defined ones. Indeed, the pairs $(\p,a)\in Q\times \Sigma$ for which no transition is defined enable the ROWJFA to perform a jump as opposed to rejecting the input as a DFA would. Hence, a ROWJFA with a complete transition function is just a DFA.\par

\textit{Sweeps} are contiguous sequences of transitions on a given input, consisting of the steps from reading or jumping over the leftmost remaining input letter to reading or jumping over the rightmost one.
If a position is jumped over, then the input symbol in that position is processed in a later sweep. 
The number of sweeps needed to process the whole input is the number of times the automaton reaches the last position of the original input word or, equivalently, one more than the maximum number of times any position is jumped over.

For an intuitive picture of sweeps, consider the computation of a ROWJFA $M$ on input $w$ as a table with rows representing the $k$ sweeps needed to process $w$ and columns representing positions in the input word. Cell $i,j$ in the table contains either a letter or a symbol representing that the letter has been \emph{read}, e.g., $\marked$. Once a letter has been marked \emph{read} and erased it stays that way, so each column is a word of the form $a^\ell \marked^{k-\ell}$ ($=a^\ell$) for some $a\in\Sigma$ and $1\leq \ell\leq k $.

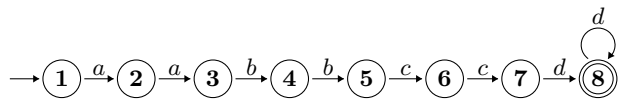
\begin{figure}[b]
\centering
\begin{center}
\begin{tikzpicture}[scale=0.1]
\tikzstyle{every node}+=[inner sep=0pt]
\draw [black] (6.5,-28.6) circle (2.5);
\draw (6.5,-28.6) node {$\st{1}$};
\draw [black] (16.4,-28.6) circle (2.5);
\draw (16.4,-28.6) node {$\st{2}$};
\draw [black] (26.6,-28.6) circle (2.5);
\draw (26.6,-28.6) node {$\st{3}$};
\draw [black] (36.8,-28.6) circle (2.5);
\draw (36.8,-28.6) node {$\st{4}$};
\draw [black] (47,-28.6) circle (2.5);
\draw (47,-28.6) node {$\st{5}$};
\draw [black] (57.4,-28.6) circle (2.5);
\draw (57.4,-28.6) node {$\st{6}$};
\draw [black] (67.5,-28.6) circle (2.5);
\draw (67.5,-28.6) node {$\st{7}$};
\draw [black] (77.8,-28.6) circle (2.5);
\draw (77.8,-28.6) node {$\st{8}$};
\draw [black] (77.8,-28.6) circle (2);
\draw [black] (-0.4,-28.6) -- (3.5,-28.6);
\fill [black] (3.5,-28.6) -- (2.7,-28.1) -- (2.7,-29.1);
\draw [black] (9.5,-28.6) -- (13.4,-28.6);
\fill [black] (13.4,-28.6) -- (12.6,-28.1) -- (12.6,-29.1);
\draw (11.45,-28.1) node [above] {$a$};
\draw [black] (19.4,-28.6) -- (23.6,-28.6);
\fill [black] (23.6,-28.6) -- (22.8,-28.1) -- (22.8,-29.1);
\draw (21.5,-28.1) node [above] {$a$};
\draw [black] (29.6,-28.6) -- (33.8,-28.6);
\fill [black] (33.8,-28.6) -- (33,-28.1) -- (33,-29.1);
\draw (31.7,-28.1) node [above] {$b$};
\draw [black] (39.8,-28.6) -- (44,-28.6);
\fill [black] (44,-28.6) -- (43.2,-28.1) -- (43.2,-29.1);
\draw (41.9,-28.1) node [above] {$b$};
\draw [black] (50,-28.6) -- (54.4,-28.6);
\fill [black] (54.4,-28.6) -- (53.6,-28.1) -- (53.6,-29.1);
\draw (52.2,-28.1) node [above] {$c$};
\draw [black] (60.4,-28.6) -- (64.5,-28.6);
\fill [black] (64.5,-28.6) -- (63.7,-28.1) -- (63.7,-29.1);
\draw (62.45,-28.1) node [above] {$c$};
\draw [black] (70.5,-28.6) -- (74.8,-28.6);
\fill [black] (74.8,-28.6) -- (74,-28.1) -- (74,-29.1);
\draw (72.65,-28.1) node [above] {$d$};
\draw [black] (76.477,-25.92) arc (234:-54:2.25);
\draw (77.8,-21.35) node [above] {$d$};
\fill [black] (79.12,-25.92) -- (80,-25.57) -- (79.19,-24.98);
\end{tikzpicture}
\end{center}
\caption{ROWJFA $M_1$ accepting all $w$ with $|w|_a=|w|_b=|w|_c=2$ and $|w|_d\geq 1$.}
\label{fig:exampleauto1}
\end{figure}

\begin{example}\label{example1}
Consider the automaton $M_1$ in Fig.~\ref{fig:exampleauto1} and the input $adcbcba$, processed in the order $aabbccd$. The ROWJFA jumps over the letter $d$ three times before processing it, hence the number of sweeps is four. Moreover, its computation table is described in Fig.~\ref{fig:comp_table}.
\end{example}

In order to be able to analyze the boundary between regular and non-regular languages accepted by the one-way jumping model, as well as to quantify the use of resources beyond the capabilities of classical DFA, when it is the case, the following complexity measure was proposed~\cite{FMW22}, which gives us the number of sweeps performed by a machine in the `worst case' for an input of length $n$.

Let $M$ be a ROWJFA and $w\in L(M)$, and let 
\[\p_0 w\vdash \p_1 w_1\vdash \p_2 w_2\vdash\dots \vdash \p_m, \mbox{ where } \p_0=s \mbox{ and } \p_m\in F,
\]
be the computation of $M$ on the input $w$. Sweep $1$ consists of $\st{p}_0w\vdash^* \st{p}_{|w|}w_{|w|}$, and we say that sweep $1$ ends in configuration $\st{p}_{|w|}w_{|w|}$. Then, for all $i\geq 1$, if sweep $i$ ends in configuration $\st{p}_{s_i} w_{s_i}$, then sweep $i+1$ is the sequence of configurations $\st{p}_{s_i}w_{s_i}\vdash^* \st{p}_{s_i+|w_{s_i}|}w_{s_i+|w_{s_i}|}$. The last sweep ends in configuration $\st{p}_m$, that is, when all input symbols have been read.
We define 
\[E(M,w)=\{\mbox{the number of sweeps performed by } M \mbox{ on } w\}.\]

When $w\notin L(M)$, then we set $E(M,w)=0$. The \emph{sweep complexity} of a machine $M$ is a function $sc_M:\N\rightarrow \N$, with $sc_M(n)$ being the maximum number of sweeps $M$ makes on processing inputs $w\in L(M)$ of length $n$, formally:
\[sc_M(n)=\max\{E(M,w) \mid w\in  \Sigma^n\}.\]
In a sense the ``most non-regular'' word (using the largest amount of non-classical resources) of each length is considered. With this in mind, we can define complexity classes in the usual manner: the class $\mathrm{SWEEP}(f(n))$ consists of languages accepted by some one-way jumping machine with sweep complexity $\mathcal{O}(f(n))$.

Observe that the sweep complexity of a machine can be defined to also take into account the sweep count of rejected words. However, this allows to `artificially' increase the sweep complexity of machines with complexity $o(n)$ without affecting regularity. Let $A$ be a machine accepting a regular language and $B$ a non-regular language with sweep complexities $f(n)$ and $g(n)$, respectively, such that $f(n)\in o(g(n))$. Then we can construct a ROWJFA accepting $aL(A)$ with sweep complexity $g(n)$ by adding a new initial state from which reading $a$ takes us to the initial state of $A$ while reading $b$ takes us to the initial state of $B$. We set all states of $B$ non-final and this way we get that on inputs starting with $b$ the machine performs $B$'s computations but never accepts anything. Moreover, $aL(A)$ is regular if and only if $L(A)$ was (see Fig.~\ref{fig:accept_rubbish}).

Each machine considered up to the point when the above measures were introduced~\cite{FMW22} had either constant or, the maximal possible, linear sweep complexity, so it seemed that there is a gap between them. Moreover, the examples with linear complexity accepted non-regular languages, while as the theorem below states, the constant complexity languages are exactly the regular languages.

\begin{theorem}[\cite{NCMA16}]\label{thm:mainReg}
ROWJFA with $\bigo(1)$ sweep complexity accept regular \mbox{languages}.
\end{theorem}

The sufficient condition above was conjectured to be also necessary for regularity in general, evidenced by the known examples at that point. 

Next, we investigate the apparent gap between constant and linear complexities and show that the presumed condition above is not necessary for regularity. Our search for machines with non-constant sweep complexity is directed by the following structural lemma, which says that such machines need to have two `complementary deficient states' in a cycle.
\begin{lemma}[\cite{FMW22}]
\label{nonconstantsweeps}
If a ROWJFA has sweep complexity $\omega(1)$ then its state diagram  has a closed walk with states $\st{p}$ and $\st{q}$, such that $\st{p}au \Ra^* \st{q} bv \Ra^* \st{p}$ for $a,b\in\Sigma$, $u,v\in \Sigma^*$ and $\st{p}$ has no transition defined for $b$, while $\st{q}$ has no transition for $a$.
\end{lemma}

\section{Regular languages with non-constant sweep complexity}\label{sec:logregular}

In this section we show that there is no sweep complexity separation between regular and non-regular languages by exhibiting automata which accept regular languages while requiring superconstant number of sweeps.

Consider first the automaton $\mathcal{B}$ with states $\{\st{1}, \st{2}, \st{3}\}$ where $\st{1}$ is initial and final, and transitions are $\{\st{1}a\rightarrow \st{2}, \st{2}a\rightarrow \st{1}, \st{1}b\rightarrow \st{3}, \st{3}b\rightarrow \st{1}\}$, described in Fig.~\ref{fig:B}.

\begin{figure}[t]
\begin{minipage}{.48\textwidth}
    \centering
    \begin{center}
\begin{tikzpicture}[scale=0.13]
\tikzstyle{every node}+=[inner sep=0pt]
\draw [black] (37.7,-24.8) circle (3);
\draw (37.7,-24.8) node {$\mathbf{1}$};
\draw [dashed,black] (22.9,-24.8) circle (3);
\draw [dashed,black] (22.9,-24.8) circle (2.4);
\draw (22.9,-24.8) node {$\mathbf{A}$};
\draw [dashed,black] (52.5,-24.8) circle (3);
\draw (52.5,-24.8) node {$\mathbf{B}$};
\draw [black] (37.7,-19) -- (37.7,-21.8);
\fill [black] (37.7,-21.8) -- (38.2,-21) -- (37.2,-21);
\draw [black] (34.7,-24.8) -- (25.9,-24.8);
\fill [black] (25.9,-24.8) -- (26.7,-25.3) -- (26.7,-24.3);
\draw (30.3,-25.3) node [below] {$a$};
\draw [black] (40.7,-24.8) -- (49.5,-24.8);
\fill [black] (49.5,-24.8) -- (48.7,-24.3) -- (48.7,-25.3);
\draw (45.1,-25.3) node [below] {$b$};
\end{tikzpicture}
\end{center}
    \caption{Artificially increasing the automaton's complexity by adding non-functional states (all final states in $\st{A}$).}
    \label{fig:accept_rubbish}
\end{minipage}
    \hfill
\begin{minipage}{.48\textwidth}
    \centering
    \begin{center}
\begin{tikzpicture}[scale=0.13]
\tikzstyle{every node}+=[inner sep=0pt]
\draw [black] (37.7,-24.8) circle (3);
\draw (37.7,-24.8) node {$\mathbf{1}$};
\draw [black] (37.7,-24.8) circle (2.4);
\draw [black] (22.9,-24.8) circle (3);
\draw (22.9,-24.8) node {$\mathbf{2}$};
\draw [black] (52.5,-24.8) circle (3);
\draw (52.5,-24.8) node {$\mathbf{3}$};
\draw [black] (37.7,-19) -- (37.7,-21.8);
\fill [black] (37.7,-21.8) -- (38.2,-21) -- (37.2,-21);
\draw [black] (34.7,-24.8) -- (25.9,-24.8);
\fill [black] (25.9,-24.8) -- (26.7,-25.3) -- (26.7,-24.3);
\draw (30.3,-25.3) node [below] {$a$};
\draw [black] (25.205,-22.898) arc (120.8641:59.1359:9.931);
\fill [black] (35.39,-22.9) -- (34.96,-22.06) -- (34.45,-22.92);
\draw (30.3,-20.99) node [above] {$a$};
\draw [black] (40.7,-24.8) -- (49.5,-24.8);
\fill [black] (49.5,-24.8) -- (48.7,-24.3) -- (48.7,-25.3);
\draw (45.1,-25.3) node [below] {$b$};
\draw [black] (40.049,-22.952) arc (119.74539:60.25461:10.18);
\fill [black] (40.05,-22.95) -- (40.99,-22.99) -- (40.5,-22.12);
\draw (45.1,-21.11) node [above] {$b$};
\end{tikzpicture}
\end{center}
    \caption{ROWJFA $\mathcal{B}$ accepts $\{w\in\{a,b\}^* \mid |w|_a \mbox{ and } |w|_b \mbox{ are even}\}$ with sweep complexity $\Theta(\log n)$.}
    \label{fig:B}
\end{minipage}
\end{figure}

\begin{proposition}\label{prop:B_reg}
    $L(\mathcal{B})$ is regular.
\end{proposition}
\begin{proof}
    We claim that $L(\mathcal{B})=\{w\in\{a,b\}^* \mid |w|_a \mbox{ and } |w|_b \mbox{ are even}\}$. This is obviously a regular language (i.e., Fig.~\ref{fig:DFAodd} where $\st{00}$ is the final state).

    The computation for a word $w$ is rejecting if it finishes in either $\st{2}$ or $\st{3}$. However, the only time that the machine ends up in state $\st{2}$ is when it reads an odd number of $a$'s, and, similarly, it ends in $\st{3}$ when it reads an odd number of $b$'s. Since both of these types of words are rejected, we conclude.
    \qed
\end{proof}

\begin{theorem}\label{prop:B_linear}
    The sweep complexity of $\mathcal{B}$ is $\Theta(\log n)$.
\end{theorem}
\begin{proof}
    Firstly, observe that in any sweep, while in $\st{1}$ or $\st{2}$, the automaton fully reads any block of $a$'s, and, similarly, while in $\st{1}$ or $\st{3}$, the automaton fully reads any block of $b$'s. 
    Thus, the number of sweeps necessary to process a word $w$ consisting of $2n$ unary blocks is never higher than that of processing the word $(ab)^{n}$. 
    Now, for the inputs $(ab)^n$ (and $(ba)^n$), starting with the first $b$ (respectively, $a$) every third symbol is jumped over while the rest is read. This means that from an arbitrary word with $k$ unary blocks, after one sweep at most 
    $\lfloor \frac{k}{3}\rfloor+1$ blocks
    remain. This immediately gives us that the machine makes at most logarithmically many sweeps. As for the other side, consider an input $w=(ab)^{6^k}$. Per the previous argument, after $i\leq $ sweeps the remaining input will be $(ab)^{\frac{6^k}{3^i}}$ or $(ba)^{\frac{6^k}{3^i}}$ depending on the parity of $i$, so the number of sweeps is at least $\log_3\frac{|w|}{2}=k$. Eventually, the input is accepted according to Proposition~\ref{prop:B_reg}, so the sweep complexity of $\mathcal{B}$ is also $\Omega(\log n)$.\qed
\end{proof}

The above results showcase the existence of ROWJFAs that accept regular languages while performing a logarithmic number of sweeps.
Next we construct of a ROWJFA that accepts a regular language while requiring a linear number of sweeps in the worst case. Consider the automaton $\mathcal{C}$ in Fig.~\ref{fig:linreg} defined as 
\[
\mathcal{C}=\{
\{\st{A0},\st{A1},\st{A2},\st{A3},\st{B1},\st{B2},\st{B3}\},
\{a,b\},
R,
\st{A0},
\{\st{B1}\}
\},
\]
where the transitions from $R$ are given by the edges in the figure.

\begin{figure}[t]
    \centering
    \begin{center}
\begin{tikzpicture}[scale=0.12]
\tikzstyle{every node}+=[inner sep=0pt]
\draw [black] (17.5,-26.6) circle (3);
\draw (17.5,-26.6) node {$\st{A0}$};
\draw [black] (27.3,-18.3) circle (3);
\draw (27.3,-18.3) node {$\st{A1}$};
\draw [black] (38,-26.6) circle (3);
\draw (38,-26.6) node {$\st{A2}$};
\draw [black] (27.3,-34.5) circle (3);
\draw (27.3,-34.5) node {$\st{A3}$};
\draw [black] (51.6,-26.6) circle (3);
\draw (51.6,-26.6) node {$\st{B1}$};
\draw [black] (51.6,-26.6) circle (2.4);
\draw [black] (64.4,-34.5) circle (3);
\draw (64.4,-34.5) node {$\st{B3}$};
\draw [black] (64.4,-18.3) circle (3);
\draw (64.4,-18.3) node {$\st{B2}$};
\draw [black] (11.8,-26.6) -- (14.5,-26.6);
\fill [black] (14.5,-26.6) -- (13.7,-26.1) -- (13.7,-27.1);
\draw [black] (19.79,-24.66) -- (25.01,-20.24);
\fill [black] (25.01,-20.24) -- (24.08,-20.37) -- (24.72,-21.14);
\draw (23.35,-22.94) node [below] {$a$};
\draw [black] (29.67,-20.14) -- (35.63,-24.76);
\fill [black] (35.63,-24.76) -- (35.3,-23.88) -- (34.69,-24.67);
\draw (31.7,-22.95) node [below] {$a$};
\draw [black] (35.59,-28.38) -- (29.71,-32.72);
\fill [black] (29.71,-32.72) -- (30.65,-32.65) -- (30.06,-31.84);
\draw (31.65,-30.05) node [above] {$b$};
\draw [black] (24.96,-32.62) -- (19.84,-28.48);
\fill [black] (19.84,-28.48) -- (20.14,-29.37) -- (20.77,-28.6);
\draw (23.41,-30.06) node [above] {$b$};
\draw [black] (30.284,-18.004) arc (92.5534:49.72998:27.696);
\fill [black] (49.42,-24.54) -- (49.13,-23.64) -- (48.49,-24.41);
\draw (41.37,-18.93) node [above] {$b$};
\draw [black] (61.437,-34.894) arc (-90.8862:-152.47853:10.168);
\fill [black] (61.44,-34.89) -- (60.64,-34.38) -- (60.63,-35.38);
\draw (55.25,-33.88) node [below] {$b$};
\draw [black] (61.85,-32.92) -- (54.15,-28.18);
\fill [black] (54.15,-28.18) -- (54.57,-29.02) -- (55.1,-28.17);
\draw (59,-30.05) node [above] {$b$};
\draw [black] (52.629,-23.792) arc (151.92357:93.99821:10.814);
\fill [black] (61.42,-18.09) -- (60.58,-17.65) -- (60.65,-18.65);
\draw (55.34,-19.31) node [above] {$a$};
\draw [black] (61.88,-19.93) -- (54.12,-24.97);
\fill [black] (54.12,-24.97) -- (55.06,-24.95) -- (54.52,-24.11);
\draw (58.94,-22.95) node [below] {$a$};
\draw [black] (49.459,-28.699) arc (-48.90124:-95.07977:25.732);
\fill [black] (49.46,-28.7) -- (48.53,-28.85) -- (49.19,-29.6);
\draw (41.33,-34.32) node [below] {$a$};
\end{tikzpicture}
\end{center}
    \caption{ROWJFA $\mathcal{C}$ accepts $\{w\in\{a,b\}^* \mid |w|_a \mbox{ and } |w|_b \mbox{ are odd}\}$ with sweep complexity $\Theta(n)$.}
    \label{fig:linreg}
\end{figure}
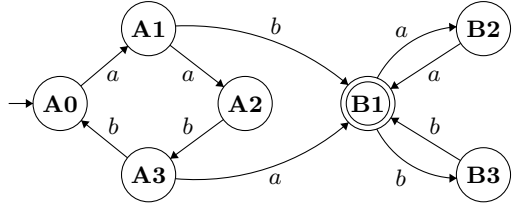

\begin{proposition}
\label{prop:C_linear}
The sweep complexity of $\mathcal{C}$ is $\Theta(n)$.
\end{proposition}
\begin{proof}
To see that the complexity is $\Omega(n)$, consider the word $a^{2n+1}b^{2n+1}$, for $n>1$. In this case, from $\st{A0}$ we go first to $\st{A2}$ where we jump over all the remaining $a$'s, then we move back to $\st{A0}$ where we jump over all the remaining $b$'s, and we are left with $a^{2n-1}b^{2n-1}$ to process. After the $n$th sweep, we are only left with $ab$ to process, which takes us from $\st{A0}$ to $\st{B1}$, and we accept.

For the $\mathcal{O}(n)$ complexity, observe that the above computation is indeed the longest possible. Once we reach $\st{B1}$ we either accept or reject a word in at most $\bigo(\log n)$ sweeps, same as in Theorem~\ref{prop:B_linear}. Of course, this part also directly follows from the fact that all ROWJFA process their inputs in $\mathcal{O}(n)$ sweeps.\qed
\end{proof}

\begin{proposition}\label{thm:C_reg}
    $L(\mathcal{C})$ is regular.
\end{proposition}
\begin{proof}
    We show that $L(\mathcal{C})=\{w\in\{a,b\}^* \mid |w|_a \mbox{ and } |w|_b \mbox{ are odd}\}$. This is obviously a regular language (i.e., Fig.~\ref{fig:DFAodd} where $\st{11}$ is the final state).

    To show that indeed $L(\mathcal{C})$ is the language containing every binary word that has odd number of $a$'s and $b$'s, first note that the right hand side automaton consisting only of the $\st{B}$-labelled states, accepts every language that has an even number of $a$'s and $b$'s, as shown by Proposition~\ref{prop:B_reg}.

    To reach $\st{B1}$ we have to read exactly one $a$ and one $b$ starting from either $\st{A0}$ or $\st{A2}$.
    Since from the start state $\st{A0}$ we can reach $\st{A0}$ or $\st{A2}$ by processing an even number of $a$'s and $b$'s, possibly with jumps, our conclusion follows.
    \qed
\end{proof}

As a consequence of Propositions~\ref{prop:C_linear} and~\ref{thm:C_reg}, we know that the class of regular languages has no upper bound in terms of sweep complexity, since the sweep complexity of any is in $\bigo(n)$. The left hand cycle in the automata $\mathcal{C}$ described in Fig.~\ref{fig:linreg} also showcases that while the conditions from Lemma~\ref{nonconstantsweeps} are necessary for non-regularity (as it requires superconstant complexity), they are not sufficient.

\section{Separation results for the language classes $\swe(\log n)$ and $\swe(n)$}\label{sec:lognonregular}

Consider the prolongable morphism $\varphi(a)=abab$, $\varphi(b)=b$ starting from the word $ab$. We get $\varphi(ab)=ababb$, $\varphi^2(ab)=\varphi(ababb)=ababbababbb$, etc. The infinite word $\phi = \lim_{n\rightarrow \infty} \varphi^n(ab)=ababbababbb\dots$ is a fixed point of $\phi$. It is easy to see that in $\phi$ all $a$'s stand alone, that is, we never have blocks of $a$'s longer than $1$, and the lengths of the blocks of $b$'s are $1,2,1,3$, and so on\footnote{The sequence $\{c(n)\}_{n=1}^{\infty}$ given by the lengths of $b$ blocks is A001511 in OEIS; its most relevant characterization for us is that $c(n) - 1$ is the number of trailing zeros in the binary expansion of $n$, since this means that $c(n)-1$ is $\log n$ for powers of $2$}. 
When applying $\varphi$, each $a$ introduces a new block of $b$'s of length $1$ and extends a block of $b$'s by one, while the number of $a$'s doubles. Thus every other block of $b$'s gets longer by one on each application of $\varphi$, because of the $a$ preceding it. A simple inductive argument shows that the last block of $b$'s in $\varphi^n(ab)$ has length $n+1$, and is preceded by $2^{n}$ occurrences of $a$'s, separated by blocks of $b$'s.

\begin{lemma}\label{lem:phi}
    Consider the morphism $\varphi:\{a,b\}^*\Ra \{a,b\}^*$ given by $\varphi(a)=abab$, $\varphi(b)=b$. The following statements hold for any $n\geq 1$:
    \begin{enumerate}[label={(\roman*)}]
        \item $\varphi^n(ab)\in (ababb^+)^+$;
        \item if $\varphi^n(ab)=ab^{k_1}\cdots ab^{k_m}$, then $\varphi^{n+1}(ab)=abab^{k_1+1}abab^{k_2+1}\cdots abab^{k_m+1}$;
        \item $\varphi^n(ab)=ab^{k_1}\cdots ab^{k_m}$, where $m=2^n$, $k_m=n+1$ and $k_{2i-1}=1$ for all $i\in\{1,\dots, 2^{n-1}\}$. 
    \end{enumerate}
\end{lemma}
\begin{proof}
    When $n=1$, then $\varphi(ab)=ababb$, so for $n=1$ all three claims hold. Suppose they hold for $n$. By $(ii)$ and $(iii)$ we have that $\varphi^{n+1}(ab)$ has the form $abab^{k_1+1}abab^{k_2+1}\cdots abab^{k_m+1}$, satisfying $(i)$ for $n+1$. Then,  
    \begin{equation*}\label{eq1}
        \begin{split}
            \varphi^{n+2}(ab) & =\varphi(abab^{k_1+1}\cdots abab^{k_m+1})=\varphi(ab)\varphi(ab^{k_1+1})\cdots \varphi(ab)\varphi(ab^{k_m+1}) \\
            & = (abab^{1+1})(abab^{k_1+2})\cdots (abab^{1+1})(abab^{k_m+2})
        \end{split}
    \end{equation*}
    
    From this we can conclude that $(ii)$ also holds for $n+1\geq 1$. Further, by the equation above we have $\varphi^{n+1}(ab)=ab^{\ell_1}\cdots ab^{\ell_{m'}}$ with $m'=2m=2\cdot 2^n=2^{n+1}$. Finally, because of $(ii)$ we also get that $\ell_{m'}=k_m+1=n+2$ and $\ell_{2i-1}=1$ for all $i\in \{1,\dots, 2^n\}$.
\qed\end{proof}

\begin{figure}[t]
\begin{minipage}{.48\textwidth}
    \centering
\begin{tikzpicture}[scale=0.13]
\tikzstyle{every node}+=[inner sep=0pt]
\draw [black] (13,-28.7) circle (3);
\draw (13,-28.7) node {$\st{1}$};
\draw [black] (30.6,-28.7) circle (3);
\draw (30.6,-28.7) node {$\st{2}$};
\draw [black] (21.8,-41.6) circle (3);
\draw [black] (21.8,-41.6) circle (2.4);
\draw (21.8,-41.6) node {$\st{3}$};
\draw [black] (16,-28.7) -- (27.6,-28.7);
\fill [black] (27.6,-28.7) -- (26.8,-28.2) -- (26.8,-29.2);
\draw (21.8,-28.2) node [above] {$a$};
\draw [black] (7.3,-28.7) -- (10,-28.7);
\fill [black] (10,-28.7) -- (9.2,-28.2) -- (9.2,-29.2);
\draw [black] (28.91,-31.18) -- (23.49,-39.12);
\fill [black] (23.49,-39.12) -- (24.35,-38.74) -- (23.53,-38.18);
\draw (26.8,-36.5) node [right] {$b$};
\draw [black] (20.11,-39.12) -- (14.69,-31.18);
\fill [black] (14.69,-31.18) -- (14.73,-32.12) -- (15.55,-31.56);
\draw (16.8,-36.5) node [left] {$b$};
\draw [black] (29.277,-26.02) arc (234:-54:2.25);
\draw (30.6,-21.45) node [above] {$a$};
\fill [black] (31.92,-26.02) -- (32.8,-25.67) -- (31.99,-25.08);
\end{tikzpicture}
    \caption{ROWJFA $\mathcal{D}$ accepts a non-regular language with $\Theta(\log n)$ sweeps.}
    \label{fig:B2}
\end{minipage}
\hfill
\begin{minipage}{.48\textwidth}
    \centering
        \begin{center}
\begin{tikzpicture}[scale=0.13]
\tikzstyle{every node}+=[inner sep=0pt]
\draw [black] (17.5,-26.6) circle (3);
\draw (17.5,-26.6) node {$\st{00}$};
\draw [black] (29,-17.3) circle (3);
\draw (29,-17.3) node {$\st{10}$};
\draw [black] (41,-26.6) circle (3);
\draw (41,-26.6) node {$\st{11}$};
\draw [black] (29,-35.7) circle (3);
\draw (29,-35.7) node {$\st{01}$};
\draw [black] (11.8,-26.6) -- (14.5,-26.6);
\fill [black] (14.5,-26.6) -- (13.7,-26.1) -- (13.7,-27.1);
\draw [black] (19.83,-24.71) -- (26.67,-19.19);
\fill [black] (26.67,-19.19) -- (25.73,-19.3) -- (26.36,-20.08);
\draw (24.2,-22.44) node [below] {$a$};
\draw [black] (19.85,-28.46) -- (26.65,-33.84);
\fill [black] (26.65,-33.84) -- (26.33,-32.95) -- (25.71,-33.73);
\draw (24.26,-30.65) node [above] {$b$};
\draw [black] (26.013,-35.58) arc (-100.4726:-156.23684:10.511);
\fill [black] (18.3,-29.48) -- (18.17,-30.41) -- (19.08,-30.01);
\draw (20.39,-33.98) node [below] {$b$};
\draw [black] (18.403,-23.749) arc (154.78156:103.14307:11.249);
\fill [black] (18.4,-23.75) -- (19.2,-23.24) -- (18.29,-22.81);
\draw (20.56,-19.3) node [above] {$a$};
\draw [black] (31.39,-33.89) -- (38.61,-28.41);
\fill [black] (38.61,-28.41) -- (37.67,-28.5) -- (38.27,-29.29);
\draw (34.05,-30.65) node [above] {$a$};
\draw [black] (39.878,-29.374) arc (-29.01233:-76.63916:12.284);
\fill [black] (31.97,-35.37) -- (32.87,-35.67) -- (32.64,-34.7);
\draw (37.51,-33.7) node [below] {$a$};
\draw [black] (31.37,-19.14) -- (38.63,-24.76);
\fill [black] (38.63,-24.76) -- (38.3,-23.88) -- (37.69,-24.67);
\draw (33.99,-22.45) node [below] {$b$};
\draw [black] (31.971,-17.66) arc (76.14526:28.30337:12.385);
\fill [black] (31.97,-17.66) -- (32.63,-18.34) -- (32.87,-17.37);
\draw (37.6,-19.4) node [above] {$b$};
\end{tikzpicture}
\end{center}
        \caption{DFA accepting words with even (for $\st{00}$ final state) or odd (for $\st{11}$ final state) number of $a$'s and $b$'s.}
        \label{fig:DFAodd}
\end{minipage}
\end{figure}

In what follows we analyze the language accepted by the automaton $\mathcal{D}=\left(\{\st{1},\st{2},\st{3}\}, \{a,b\}, \{\st{1}a\rightarrow \st{2}, \st{2}a\rightarrow \st{2}, \st{2}b\rightarrow \st{3}, \st{3}b\rightarrow \st{1}\}, \st{1}, \{\st{3}\}\right)$, described in Fig.~\ref{fig:B2}.

\begin{lemma}\label{lem:phiaccept}
    For any $n\geq 0$, the ROWJFA $\mathcal{D}$ accepts $\varphi^n(ab)$ in $n+1$ sweeps.
\end{lemma}
\begin{proof}
We show that the machine accepts $\varphi^n(ab)$, for any $n\geq 0$. 
From state $\st{1}$ after reading/jumping through a factor of the form $ababb^+$ the automaton gets back to state $\st{1}$. In fact, $\st{1}abab^kw \vdash^* \st{1}wab^{k-1}$, for any $k\geq 1$, so in one sweep the factor $abab^k$ is reduced to $ab^{k-1}$. From Lemma~\ref{lem:phi} we can see that we can write $\varphi^{n+1}(ab)=abab^{k_1+1}abab^{k_2+1}\cdots abab^{k_m+1}$, which means that one sweep of $\mathcal{D}$ acts as the inverse of $\varphi$ on those words when starting from state $\st{1}$, that is, 
\[\st{1}\varphi^{n+1}(ab)=\st{1}abab^{k_1+1}abab^{k_2+1}\cdots abab^{k_m+1} \vdash^* \st{1} ab^{k_1}ab^{k_2}\cdots ab^{k_m}= \st{1} \varphi^{n}(ab).\]
This means that in $n$ sweeps the machine reduces $\varphi^n(ab)$ to $\varphi^0(ab)$.
Finally, for $n=0$, we have $\varphi^0(ab)=ab$, which is accepted by $\mathcal{D}$ in a single sweep.\qed
\end{proof}

\begin{lemma}\label{lem:B2_non_reg}
    The ROWJFA $\mathcal{D}$ accepts a non-regular language.
\end{lemma}
\begin{proof}
By Lemma~\ref{lem:phiaccept} we know that for any $n$ the machine accepts $\varphi^n(ab)$, which means that for arbitrarily long unary factors consisting of $b$'s, there is some word in $L(\mathcal{D})$ having such a factor as a suffix. Our strategy is to first establish a non-linear relation between the length of those unary factors and the length of the preceding factors in all words accepted by $\mathcal{D}$. Then, by a pumping argument we show that a classical finite automaton cannot verify such a non-linear relation, therefore $L(\mathcal{D})$ cannot be regular.

\noindent{\bf Claim 1.} Words of the form $wb^n$ are only accepted if $|w|\in \Omega(2^{\frac{n}{2}})$.\\
\emph{Proof of Claim 1:}
In any sweep, any block of $a$'s which the automaton starts to read is read and erased completely through a sequence of transitions $\st{1}a^kbu \Ra^* \st{2} bu$. For the automaton to jump over a block of $a$'s, it needs to arrive to its start in state $\st{3}$. Then it jumps over it to the next $b$, after which it starts and reads completely the following block of $a$'s, as described earlier. This means that the machine can never jump over two consecutive blocks of $a$'s. From here we get that if at the beginning of the sweep the number of $a$ blocks was $\ell$, then after the sweep it is at most $\lfloor \frac{\ell}{2}\rfloor+1$.

Furthermore, in each sweep, each block of $b$'s is reduced by at most $2$. This means that the automaton needs at least $\frac{n}{2}$ sweeps to read a block $b^n$, in each of which it reduces the number of $a$ blocks by half (or more). Thus we can conclude that in order to accept a word with a suffix $b^n$, we have to start out with at least $2^{\frac{n}{2}}$ blocks of $a$'s preceding it.\hfill$\nabla$

\noindent{\bf Claim 2.} No finite automaton can accept $L(\mathcal{D})$.\\
\emph{Proof of Claim 2:}
Suppose the opposite, i.e., that there exists some complete DFA $\mathcal{F}$ having $N$ states such that $L(\mathcal{F})=L(\mathcal{D})$. We know that there are words in the language with arbitrarily long suffixes of $b$'s, so there is a $wb^m\in L(\mathcal{F})$ for some word $w$ and exponent $m>N$. By a usual pumping argument, this means that there exists some $\ell$ with $0<\ell<N$ such that $wb^{m+i\cdot \ell}\in L(\mathcal{F})$ for any $i\geq 0$. However, for a large enough $i$ this contradicts Claim 1, as the block of $b$'s can outgrow any upper bound in terms of the length of $|w|$.\hfill$\nabla$

Our result follows as a result of Claims~1 and~2.\qed
\end{proof}

\begin{lemma}
\label{lem:B2_log}
    The sweep complexity of $\mathcal{D}$ is $\Theta(\log n)$.
\end{lemma}
\begin{proof}
    As $|\varphi^n (ab)| = 2^{n+1} + 2^n - 1$,
    by Lemma~\ref{lem:phiaccept} we have that the sweep complexity of $\mathcal{D}$ is $\Omega(\log n)$, so what remains to show is that it is also $\bigo(\log n)$.
    
    We first note that within a sweep all blocks of $a$'s separated by $bb$ are fully processed (including any prefix of $a$'s), while for any symbols $a$ that were jumped over, the entire block that they were part of it was jumped over. 
    Following the argument in the proof of Claim 1 of Lemma~\ref{lem:B2_non_reg}, in each sweep the number of blocks of $a$'s is reduced by at least half, which means that after $\bigo(\log n)$ sweeps there are no more blocks of $a$ on the tape. Then, the machine either accepts in one sweep or it rejects the input. This leads to our conclusion.\qed
\end{proof}

The results of Lemmas~\ref{lem:B2_non_reg} and~\ref{lem:B2_log} mean that we have separation between $\swe(1)$ and $\swe(\log n)$.
\begin{theorem}\label{thm:logseparate}
$\swe(1)\subsetneq \swe(\log n)$
\end{theorem}
\begin{proof}
    Lemma~\ref{lem:B2_log} says $L(\mathcal{D})\in \swe(\log n)$. By Theorem~\ref{thm:mainReg} we know that $\swe(1)$ is included in the class of regular languages. Finally, by Lemma~\ref{lem:B2_non_reg} we have that $L(\mathcal{D})$ is not regular which means that $L(\mathcal{D})\notin \swe(1)$.\qed
\end{proof}

\begin{lemma}\label{thm:ablin}
Any automaton which accepts $L_{ab}=\{w\in \{a,b\}^* \mid |w|_a=|w|_b \}$ has sweep complexity $\Theta(n)$.
\end{lemma}
\begin{proof}
    We know that every machine has sweep complexity $\mathcal{O}(n)$, so it is enough to show that it is not possible to accept $L_{ab}$ with sublinear sweep complexity. For that we assume that such an automaton, say $\mathcal{F}=(Q,\Sigma,R,s,F)$ exists, and derive a contradiction.

    If $\mathcal{F}$ had linear sweep complexity, then it could have computations on infinitely many inputs in which all sweeps process a constant number of symbols. However, with sublinear complexity we get that for any constant $C$ and for all long enough inputs $w\in L_{ab}$, during the processing of $w$ at least one sweep reads more than $C$ symbols. We also know that $a^nb^n\in L_{ab}$ for any $n\geq 0$. Let $C=2|Q|$ where $|Q|$ is the number of states of $\mathcal{F}$ and consider an input $w=a^mb^m$ with $m$ large enough that the machine reads more than $C$ symbols in some sweep while processing $w$. The remaining input at the beginning of that sweep is $a^kb^\ell$ for some $k,\ell$ such that $k+\ell > C$. During the sweep the machine reads $a^{k'}b^{\ell'}$ where $k'+\ell'>C$. This means that either $k'>|Q|$ or $\ell'>|Q|$. Without loss of generality we can assume $k'>|Q|$. This gives us that while reading $a^{k'}$ the automaton must visit some state $\st{p}$ at least twice while reading only $a$'s, so we get that $\st{p}a^r\Ra^* \st{p}$ for some $r>0$. But then, by a usual pumping argument the machine also needs to accept $a^{n+r}b^n\notin L_{ab}$ contradicting our assumption that $L(\mathcal{F})=L_{ab}$ and concluding the proof.\qed
\end{proof}

\begin{theorem}
    For any $f:\N\Ra \N$ with $f(n)\in o(n)$ we have $\swe(f(n))\subsetneq \swe(n)$.
\end{theorem}
\begin{proof}
    By Lemma~\ref{thm:ablin} we know that $L_{ab}\notin \swe(f(n))$ for any sublinear function $f(n)$. The two-state automaton $\mathcal{A}$ accepts the language with sweep complexity $\Theta(n)$. This is easy to see when considering the worst-case inputs of the form $a^nb^n$ for $n\geq 0$.\qed
\end{proof}

\section{Concluding remarks}
Apart from the complexity considerations listed below we think the proof of Lemma~\ref{lem:B2_non_reg} contains a detail worth emphasizing: the automaton can verify a logarithmic/exponential relation between two factors of suitably chosen inputs! We found this very surprising since we still basically deal with DFA which cannot store information and cannot `choose' which symbols to read or jump over\footnote{Iterated uniform finite transducers can also verify such relationships, albeit their computing power is much stronger.~\cite{iteratedfst2022}}. 

We presented automata for all pairings of regular and non-regular languages with logarithmic and linear worst case sweep complexity. This way we disproved the conjecture on the constant sweep requirement for regularity~\cite{NCMA16} and answered several questions regarding sweep complexity posed in~\cite{FMW22}:
\begin{enumerate}
    \item\label{q1} Is the language of each machine with $\omega(1)$ complexity non-regular? NO, by Section~\ref{sec:logregular}.
    \item\label{q2} Is there a machine with sweep complexity between constant and linear, that is, $\omega(1)$ and $o(n)$? YES, by Theorem~\ref{prop:B_linear} (and Lemma~\ref{lem:B2_log}).
    \item\label{q3} Is there a \emph{language} with sweep complexity between constant and linear, that is, all machines accepting it have superconstant complexity and at least one has sublinear? YES, by Theorem~\ref{thm:logseparate}.
    \item\label{q4} Is there an upper bound in terms of sweep complexity on machines accepting regular languages? NO, by Propositions~\ref{prop:C_linear} and~\ref{thm:C_reg}.
    \item\label{q6} Are machines less complex in the case of binary alphabets, given that the complementary deficient pairs of Lemma~\ref{nonconstantsweeps} are predetermined? NO, illustrated by the fact that all results have been obtained over a binary alphabet.
\end{enumerate}

These coarser forms of Questions~\ref{q2} and~\ref{q3} have been answered here, but for a complete picture one would want to know whether there exist machines with arbitrary (constructible) sublinear complexity and its equivalent for languages. The most obvious choices for such a study would probably be complexities $\Theta(\log^k n)$ and $\Theta(n^\epsilon)$, for constants $k>1$ and $\epsilon<1$. Another angle related to Question~\ref{q4} is to study the lower bound of non-regularity: logarithmic complexity can produce non-regular languages, but can we do it with less of this `non-regular' resource?
In the case of Question~\ref{q6}, our answer may be refined, as there may by some sublinear $f(n)$ such that the machines of $\Theta(f(n))$ complexity all accept regular or all accept non-regular languages, although we have not seen anything that indicates such perplexing behaviour.

Another interesting direction relates to our original motivation in looking at the complexity of these automata, deciding regularity. The question more generally becomes, is it decidable given a machine or language and a function $f(n)$, whether the machine/language has $\Theta(f(n))$ complexity (or its one-sided variants with $\mathcal{O}$ and $\Omega$)? We suspect that the answer is yes at least in the case of constant and linear functions but have no idea about the logarithmic and more complicated cases.

\bibliographystyle{splncs04}
\bibliography{rev_onewayjumping}

\end{document}